\documentclass{llncs}
\usepackage{makeidx}  
\usepackage[ruled,linesnumbered,commentsnumbered]{algorithm2e}

\usepackage{color,soul}
\usepackage{amssymb}
\usepackage{graphicx}
\usepackage{color}
\usepackage{soul}

\newcommand{\listpaths}{\ensuremath{\mathtt{list\_paths}}}
\newcommand{\listpathsiter}{\ensuremath{\mathtt{list\_paths\_iterative}}}
\newcommand{\lcp}{\ensuremath{\mathtt{lcp}}}
\newcommand{\ie}{\textit{i.e. }}

\begin{document}

\title{Efficiently listing bounded length $st$-paths\thanks{GS and MFS
    were partially supported by the ERC programme FP7/2007-2013 / ERC
    grant agreement no. [247073]10, and the French project
    ANR-12-BS02-0008 (Colib'read).}}

\titlerunning{Listing bounded length $st$-paths}  
%

\author{Romeo Rizzi\inst{1} \and Gustavo Sacomoto\inst{2,3} \and Marie-France Sagot\inst{2,3}}
%
\tocauthor{}
\institute{
  Dipartimento di Informatica, Universit\`a di Verona, Italy
  \and
  Universit\'e de Lyon, F-69000 Lyon; Universit\'e Lyon 1; CNRS, UMR5558,
  Laboratoire de Biom\'etrie et Biologie \'Evolutive, F-69622 Villeurbanne, France 
  \and  
  INRIA Grenoble Rh\^one-Alpes, France 
}

\maketitle    

\begin{abstract}
  The problem of listing the $K$ shortest simple (loopless) $st$-paths
  in a graph has been studied since the early 1960s. For a
  non-negatively weighted graph with $n$ vertices and $m$ edges, the
  most efficient solution is an $O(K(mn + n^2 \log n))$ algorithm for
  directed graphs by Yen and Lawler [Management Science, 1971 and
    1972], and an $O(K(m+n \log n))$ algorithm for the undirected
  version by Katoh \textit{et al.} [Networks, 1982], both using $O(Kn
  + m)$ space. In this work, we consider a different parameterization
  for this problem: instead of bounding the number of $st$-paths
  output, we bound their length. For the bounded length
  parameterization, we propose new non-trivial algorithms matching the
  time complexity of the classic algorithms but using only $O(m+n)$
  space. Moreover, we provide a unified framework such that the
  solutions to both parameterizations -- the classic $K$-shortest and
  the new length-bounded paths -- can be seen as two different
  traversals of a same tree, a Dijkstra-like and a DFS-like traversal,
  respectively.
\end{abstract}


\section{Introduction}
The $K$-shortest simple paths problem has been studied for more than
50 years (see the references in \cite{Dreyfus69}). The first efficient
algorithm for this problem in directed graphs with non-negative
weights only appeared 10 years later independently by Yen~\cite{Yen71}
and Lawler~\cite{Lawler72}. Given a non-negatively weighted directed
graph $G = (V,E)$ with $n = |V|$ vertices and $m = |E|$ edges, using
modern data structures \cite{Tarjan90}, their algorithm lists the $K$
distinct shortest \emph{simple} $st$-paths by non-decreasing order of
the their lengths in $O(K (mn + n^2 \log n))$ time. For undirected
graphs, Katoh \textit{et al.} \cite{Katoh82} gave an improved $O(K(m+n
\log n))$ algorithm. Both algorithms use $O(Kn + m)$ memory.

The best known algorithm for directed \emph{unweighted} graphs is an
$\widetilde{O}(Km\sqrt{n})$ randomized algorithm \cite{Roditty05},
where $\widetilde{O}(f(n))$ is a shorthand for $O(f(n) \log^k n)$.  In
a different direction, Roditty~\cite{Roditty07} noticed that the
$K$-shortest simple paths can be efficiently approximated. Building
upon his work, Bernstein~\cite{Bernstein10} presented an
$\widetilde{O}(Km/\epsilon)$ time algorithm for a $(1 +
\epsilon)$-approximation. Moreover, Eppstein~\cite{Eppstein99} showed
that if the paths are allowed to repeat vertices, \ie they are not
\emph{simple}, then the problem can be solved in $O(K + m + n \log n)$
time. However, when the paths are \emph{simple} and to be computed
exactly, no improvement has been made on Yen and Lawler's for directed
graphs or Katoh's algorithm for undirected graphs. The main bottleneck
of these algorithms is their memory consumption.

Here, we consider the problem of listing all $st$-paths with length at
most $\alpha$. This is a different parameterization of the
$K$-shortest path problem, where we impose an upper-bound on the
length of the output paths instead of their number. This is a natural
variant of the $K$-shortest path problem. There are situations where
it is necessary to consider all paths that are a given percentage of
the optimal (\textit{e.g.}  \cite{Bohmova13}). Moreover, the bounded
length problem is \emph{almost} a particular case of the $K$-shortest
path problem. Given any solution to the $K$-shortest path problem,
such that the $st$-paths are generated one at a time in non-decreasing
length order, we can use the following simple approach to solve the
$\alpha$-bounded length variant: choose a sufficiently large $K$ and
halt the enumeration when the length of the paths is larger than
$\alpha$.  The main disadvantage of this algorithm is its space
complexity which is proportional to the number of paths output hence,
in the worst case, exponential in the size of the graph.

Our first and main contribution are new polynomial delay algorithms to
list $st$-paths with length at most $\alpha$ matching the time
complexity (per path) of Yen and Lawler's algorithm for directed
graphs (Section~\ref{sec:simple_alg}) and Katoh's for undirected
graphs (Section~\ref{sec:weighted:improved}), but using only $O(n+m)$
internal memory. This represents an exponential improvement in memory
consumption.

The main differences between the classic solutions to the $K$-shortest
paths problem and our solutions to the $\alpha$-bounded paths problem
are the order in which the solutions are output and the memory
complexity of the algorithms.

Our second contribution is thus a unified framework where both
problems can be represented in such a way that those differences arise
in a natural manner (Section~\ref{sec:simple_alg}). Intuitively, we
show that both families of algorithmic solutions correspond to two
different traversals of a \emph{same} rooted tree: a Dijkstra-like
traversal for the $K$-shortest and a DFS-like traversal for the
$\alpha$-bounded paths.

\section{Preliminaries}
Given a directed graph $G = (V,E)$ with $n = |V|$ vertices and $m =
|E|$ arcs, the in and out-neighborhoods of $v \in V$ are denoted by
$N^-(v)$ and $N^+(v)$, respectively. Given a (directed or undirected)
graph $G$ with weights $w : E \mapsto \mathbb{Q}$, the weight, or
\emph{length}, of a path $\pi$ is $\sum_{(u,v) \in \pi} w(u,v)$ and is
denoted by $w(\pi)$. We say that a path $\pi$ is
\emph{$\alpha$-bounded} if its \emph{length} satisfies $w(p) \leq
\alpha$ and $\alpha \in \mathbb{Q}$; in the particular case of unit
weights (\ie of unweighted graphs), we say that $p$ is
\emph{$k$-bounded} if $w(p) \leq k$ with $k \in \mathbb{Z}_{\geq
  0}$. A listing algorithm is \emph{polynomial delay} if it generates
the solutions, one after the other in some order, and the time elapsed
until the first is output, and thereafter the time elapsed (delay)
between any two consecutive solutions, is bounded by a polynomial in
the input size \cite{Johnson88}. The general problem which we are
concerned in this work is listing $\alpha$-bounded $st$-paths in $G$.

\begin{problem}[Listing $\alpha$-bounded $st$-paths] \label{prob:list}
  Given a weighted directed graph $G = (V,E)$, two vertices $s,t \in
  V$, and an upper bound $\alpha \in \mathbb{Q}$, output all
  $\alpha$-bounded $st$-paths.
\end{problem}

Clearly, any solution to the $K$-shortest path problem is also a
solution to Problem~\ref{prob:list}, with the same (total/delay) time
and space complexities. Thus Problem~\ref{prob:list} is no harder than
the classic $K$-shortest path problem. 

We assume all directed graphs are weakly connected and all undirected
graphs are connected, hence $m \geq n-1$. Moreover, we assume
hereafter the weights are non-negative. We remark however that a
weaker assumption suffices to the applicability of our
algorithms. Indeed, it is a well known fact that, when the graph $G$
and the weights $w : E \mapsto \mathbb{Q}$ are such that no cycle is
negative, then, using Johnson's reweighting strategy~\cite{Johnson77},
we can compute non-negative weights $w'$ such that, for some constant
$C$, we have that $w'(\pi) = w(\pi) + C$ for any $st$-path $\pi$.
This reweighting can be done in $O(mn)$ preprocessing steps.

\section{An $O(mn + n^2 \log n)$-delay algorithm} \label{sec:simple_alg}

In this section, we present an $O(mn + n^2 \log n)$-delay algorithm to
list all $st$-paths with length at most $\alpha$ in a weighted
directed graph $G$. Thus matching the time complexity (per path) of
Yen and Lawler's algorithm, while using only space linear in the
\emph{input} size.

The new algorithm, inspired by the binary partition
method~\cite{Birmele13,Rizzi14}, recursively partitions the solution
space at every call until the considered subspace is a singleton
(contains only one solution) and in that case outputs the
corresponding solution. In order to have an efficient algorithm is
important to explore only non-empty partitions. Moreover, it should be
stressed that the order in which the solutions are output is fixed,
but arbitrary.

Let us describe the partition scheme. Let
$\mathcal{P}_{\alpha}(s,t,G)$ be the set of all $\alpha$-bounded paths
from $s$ to $t$ in $G$, and $(x,s) \cdot \mathcal{P}_{\alpha}(s,t,G)$
denote the concatenation of $(x,s)$ to each path of
$\mathcal{P}_{\alpha}(s,t,G)$. Assuming $s \neq t$, we have that
\begin{equation} \label{eq:path:partition}
\mathcal{P}_{\alpha}(s,t,G) = \bigcup_{v \in N^+(s)} (s,v) \cdot
\mathcal{P}_{\alpha'} (v,t,G-s),
\end{equation}
where $\alpha' = \alpha - w(s,v)$. In words, the set of paths from $s$
to $t$ can be \emph{partitioned} into the disjoint union of $(s,v)
\cdot \mathcal{P}_{\alpha'} (v, t,G - s)$, the sets of paths beginning
with an arc $(s,v)$, for each $v \in N^+(s)$. Indeed, since $s \neq
t$, every path in $\mathcal{P}_{\alpha} (s, t, G)$ necessarily begins
with an arc $(s,v)$, where $v \in N^+(s)$.

Algorithm~\ref{alg:simple} implements this recursive partition
strategy. The solutions are only output in the leaves of the recursion
tree (line~\ref{alg:simple:output}), where the partition is always a
singleton. Moreover, in order to guarantee that every leaf in the
recursion tree outputs one solution, we have to test if
$\mathcal{P}_{\alpha'} (v, t, G - u)$, where $\alpha' = \alpha -
w(u,v)$, is not empty before the recursive call
(line~\ref{alg:simple:test}). This set is not empty if and only if the
weight of the shortest path from $v$ to $t$ in $G-u$ is at most
$\alpha'$, \ie $d_{G-u}(v,t) \leq \alpha' = \alpha - w(u,v)$. Hence,
to perform this test it is enough to compute all the distances from
$t$ in the graph $G^R - u$, where $G^R$ is the graph $G$ with all arcs
reversed.

Consider a generic execution of Algorithm~\ref{alg:simple} for a graph
$G$, vertices $s,t \in V$ and an upper bound $\alpha$. We can
represent this execution by a rooted tree $\mathcal{T}$, \ie the
recursion tree, where each node corresponds to a call with arguments
$\langle u,t,\alpha, \pi_{su}, G' \rangle$. The children of a given
node (call) in $\mathcal{T}$ are the recursive calls with arguments
$\langle v,t,\alpha', \pi_{su} (u,v), G' - u \rangle$ of
line~\ref{alg:simple:rec}. This tree plays an important role in the
unified framework of Section~\ref{sec:weighted:ordered}.

\begin{lemma} \label{lem:rec:tree}
  The recursion tree $\mathcal{T}$ has the following properties:
  \begin{enumerate}
    \item The leaves of $\mathcal{T}$ are in one-to-one correspondence
      with the paths in $\mathcal{P}_{\alpha}(s,t,G)$.
    \item The leaves in the subtree rooted on a node $\langle
      u,t,\alpha, \pi_{su}, G' \rangle$ correspond to the paths in
      $\pi_{su} \cdot \mathcal{P}_{\alpha'}(u,t,G')$.
    \item The height of $\mathcal{T}$ is bounded by $n$.
  \end{enumerate}
\end{lemma}

\begin{algorithm} 
\caption{$\listpaths(u,t, \alpha, \pi_{su}, G)$} \label{alg:simple}
\If{$u = t$}{ 
  output($\pi_{su}$) \\ \label{alg:simple:output}
  \bf return 
}
compute the distances from $t$ in $G^R - u$ \label{alg:simple:dist}\\
\For{$v \in N^+(u)$}{
  \If{$d(v,t) \leq \alpha - w(u,v)$}{ \label{alg:simple:test} 
    $\listpaths(v,t, \alpha - w(u,v), \pi_{su} \cdot (u,v), G - u)$ \label{alg:simple:rec} 
  } 
}
\end{algorithm}

The correctness of Algorithm~\ref{alg:simple} follows directly from
the relation given in Eq.~\ref{eq:path:partition} and the correctness
of the tests of line~\ref{alg:simple:test}. 

Let us now analyze its running time. The cost of a node in
$\mathcal{T}$ is the time spent by the operations inside the
corresponding call, without including its recursive calls. This cost
is dominated by the tests of line~\ref{alg:simple:test}. They are
performed in $O(1)$ time by pre-computing the distances from $t$ to
all vertices in the reverse graph $G^R - u$
(line~\ref{alg:simple:dist}). This takes $O(t(n,m))$ time, where
$t(n,m)$ is the cost of a single source shortest path computation. By
Lemma~\ref{lem:rec:tree} the height of $\mathcal{T}$ is bounded by
$n$, so the path between any two leaves (solutions) in the recursion
tree has at most $2n$ nodes. Thus, the time elapsed between two
solutions being output is $O(nt(n,m))$. Moreover, the algorithm uses
$O(m)$ space, since each recursive call has to store only the
difference with the its parent graph.  Recall that each solution is
immediately output (line~\ref{alg:simple:output}), not stored by the
algorithm.

\begin{theorem} \label{thm:simple_complexity}
  Algorithm~\ref{alg:simple} has delay $O(n t(n,m))$, where $t(n,m)$
  is the cost of a single source shortest path computation, and uses
  $O(m)$ space. 
\end{theorem}

For unweighted (directed and undirected) graphs, the single source
shortest paths can be computed using breadth-first search (BFS)
running in $O(m)$ time, so Theorem~\ref{thm:simple_complexity}
guarantees an $O(km)$ delay to list all $k$-bounded $st$-paths, since
the height of the recursion tree is bounded by $k$ instead of
$n$. More generally, the single source shortest paths can be computed
using Dijkstra's algorithm in $O(m + n \log n)$ time (we are assuming
non-negative weights), resulting in an $O(nm + n^2 \log n)$ delay.

\section{An improved algorithm for undirected graphs} \label{sec:weighted:improved}

The total time complexity of Algorithm~\ref{alg:simple} is equal to
the delay times the number of solutions, \ie $O(nt(n,m) \gamma)$,
where $\gamma = |\mathcal{P}_{\alpha}(s,t,G)|$ is the number of
$\alpha$-bounded $st$-paths. We now improve its total time complexity
from $O(nt(n,m) \gamma)$ to $O((m+t(n,m)) \gamma)$ in the case of
weighted \emph{undirected} graphs. On average the algorithm spends
$O(m+t(n,m))$ per solution (amortized delay), thus matching the time
complexity (per path) of Katoh's algorithm. The (worst-case) delay,
however, remains the same as Algorithm~\ref{alg:simple}.

The main idea to improve the complexity of Algorithm~\ref{alg:simple}
is to explore the structure of the set of paths
$\mathcal{P}_{\alpha}(s,t,G)$ to reduce the number of nodes in the
recursion tree. We avoid redundant partition steps by guaranteeing
that every node in the recursion tree has at least two children. More
precisely, at every call, we identify the longest common prefix of
$\mathcal{P}_{\alpha}(s,t,G)$, \ie the longest (considering the number
of edges) path $\pi_{ss'}$ such that $\mathcal{P}_{\alpha}(s,t,G) =
\pi_{ss'} \cdot \mathcal{P}_{\alpha}(s',t,G)$, and append it to the
current path prefix being considered in the recursive call. The
intuition here is that by doing so we identify and ``merge'' all the
consecutive single-child nodes in the recursion tree, thus
guaranteeing that the remaining nodes have at least two children.

The pseudocode for this algorithm is very similar to
Algorithm~\ref{alg:simple} and, for the sake of completeness, is given
in Algorithm~\ref{alg:improved}. We postpone the description of the
$\lcp(u, t, \alpha, G)$ function to the next section, along with a
discussion about the difficulties to extend it to directed graphs.

\begin{algorithm} 
\caption{$\listpaths(u,t, \alpha, \pi_{su}, G)$} \label{alg:improved}
$\pi_{uu'}$ = $\lcp(u, t, \alpha, G)$ \\
\uIf{$u' = t$}{ 
  output($\pi_{su}\pi_{uu'}$) \\ 
  \bf return 
}
\Else{
  compute a shortest path tree $T'_t$ from $t$ in $G^R - \pi_{uu'}$ \\
  \For{$v \in N(u')$}{
    \If{$d(v,t) + w(u,v) \leq \alpha$}{ 
      $\listpaths(v,t, \alpha - w(\pi_{uu'}) - w(u',v), \pi_{su} \cdot \pi_{uu'} \cdot (u',v), G - \pi_{uu'})$ 
    } 
  }
}
\end{algorithm}

The correctness of Algorithm~\ref{alg:improved} follows directly from
the correctness of Algorithm~\ref{alg:simple}. The space used is the
same of Algorithm~\ref{alg:simple}, provided that $\lcp(u, t, \alpha,
G)$ uses linear space, which, as we show in the next section, is
indeed the case (Theorem~\ref{teo:lcp}).

Let us now analyze the total complexity of
Algorithm~\ref{alg:improved} as a function of the input size and of
$\gamma$, the number of $\alpha$-bounded $st$-paths.  Let $R$ be the
recursion tree of Algorithm~\ref{alg:improved} and $T(r)$ the cost of
a given node $r \in R$. The total cost of the algorithm can be split
in two parts, which we later bound individually, in the following way:

\begin{equation}
   \sum_{r \in R} T(r) = \sum_{r: internal} T(r) + \sum_{r: leaf} T(r). \label{eq:total_cost}
\end{equation}

We have that $\sum_{r: leaf} T(r) = O((m + t(m,n))\gamma)$, since
leaves and solutions are in one-to-one correspondence and the cost for
each leaf is dominated by the cost of $\lcp(u, t, \alpha, G)$, that is
$O(m + t(m,n))$ (Theorem~\ref{teo:lcp}). Now, we have that every
internal node of the recursion has at least two children, otherwise
$\pi_{uu'}$ would not be the longest common prefix of
$\mathcal{P}_{\alpha}(u,t,G)$. Thus, $\sum_{r: internal} T(r) = O((m +
t(m,n))\gamma)$ since in any tree the number of branching nodes is at
most the number of leaves, and the cost of each internal node is
dominated by the $O(m + t(m,n))$ cost of the longest prefix
computation. Therefore, the total complexity of
Algorithm~\ref{alg:improved} is $O((m + t(n,m))\gamma)$. This
completes the proof of Theorem~\ref{thm:cost_improved}.

\begin{theorem} \label{thm:cost_improved}
  Algorithm~\ref{alg:improved} outputs all $\alpha$-bounded $st$-paths
  in $O((m+t(n,m))\gamma)$ time using $O(m)$ space.
\end{theorem}

This means that for unweighted graphs, it is possible to list all
$k$-bounded $st$-paths in $O(m)$ time per path.  In addition, for
weighted graphs, it is possible to list all $\alpha$-bounded
$st$-paths in $O(m + n \log n)$ time per path. 

\subsection{Computing the longest common prefix of $\mathcal{P}_{\alpha}(s,t,G)$}

The problem of computing the longest common prefix of
$\mathcal{P}_{\alpha}(s,t,G)$ can be seen as a special case of the
\emph{replacement paths problem}~\cite{Hershberger01}. Let $\pi$ be a
shortest $st$-path in $G$. In this problem we want to compute, for
each edge $e$ on $\pi$, the shortest $st$-path that avoids $e$. Given
a solution to the replacement path problem we can compute the longest
common prefix of $\mathcal{P}_{\alpha}(s,t,G)$ using the following
procedure. For each edge $e$ along the path $\pi$, check whether the
shortest $st$-path avoiding $e$ is shorter than $\alpha$. There is an
$O(m + n \log n)$ algorithm to compute the replacement path in
undirected graphs~\cite{Malik89}, but for directed graphs the best
solutions is a trivial $O(nm + n^2 \log n)$ algorithm.

In this section, we present an alternative, arguably simpler,
algorithm to compute the longest common prefix of the set of
$\alpha$-paths from $s$ to $t$, completing the description of
Algorithm~\ref{alg:improved}. The naive algorithm for this problem
runs in $O(nt(n,m))$ time, so that using it in Algorithm~\ref{alg:improved}
would not improve the total complexity compared to
Algorithm~\ref{alg:simple}. Basically, the naive algorithm computes a
shortest path $\pi_{st}$ and then for each prefix in increasing order
of length tests if there are at least two distinct extensions each
with total weight less than $\alpha$.  In order to test the
extensions, for each prefix $\pi_{su}$, we recompute the distances
from $t$ in the graph $G - \pi_{su}$, thus performing $n$ shortest
path tree computations ($k$ computations in the unweighted case) in
the worst case.

Algorithm~\ref{alg:lcp} improves the naive algorithm by avoiding those
recomputations.  However, before entering the description of
Algorithm~\ref{alg:lcp}, we need a better characterization of the
structure of the longest common prefix of
$\mathcal{P}_{\alpha}(s,t,G)$. Lemma~\ref{lem:lcp} gives this. It does
so by considering a shortest path tree rooted at $s$, denoted by
$T_s$.  Recall that $T_s$ is a subgraph of $G$ and induces a partition
of the edges of $G$ into tree edges and non-tree edges. In this tree,
the longest common prefix of $\mathcal{P}_{\alpha}(s,t,G)$ is a prefix
of the tree path from the root $s$ to $t$. Additionally, any $st$-path
in $G$, excluding the tree path, necessarily passes through at least
one non-tree edge.  The lemma characterizes the longest common prefix
in terms of the non-tree edges from the subtrees rooted at siblings of
the vertices in the tree path from $s$ to $t$. For instance, in
Fig.~\ref{fig:weighted:lemma}(b) the common prefix $\pi_{su}$ can be
extended to $\pi_{su} \cdot (u,v)$ only if there is no
$\alpha$-bounded path that passes through the subtree $T_w$ and a
non-tree edge $(x,z)$, where $v$ belongs to tree path from $s$ to $t$
and $w$ is one of its siblings.

\begin{figure}[htbp]
  \centering \def\svgwidth{0.8\linewidth}
  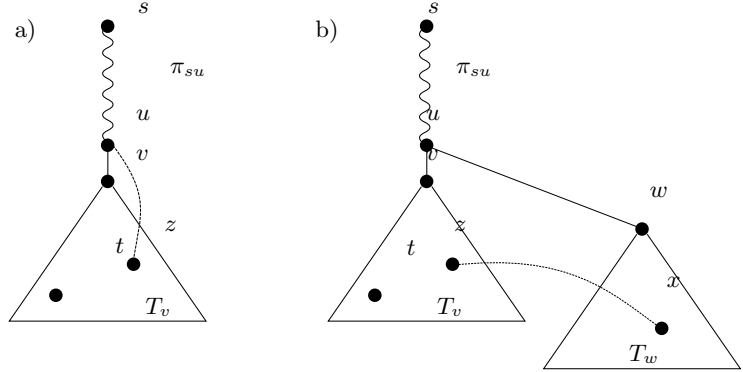 
  \caption{The common prefix $\pi_{su}$ of
    $\mathcal{P}_{\alpha}(s,t,G)$ can always be extended into an
    $st$-path using the tree path of $T_s$ from $u$ to $t$. The path
    $\pi_{su}$ is the longest common prefix if and only if it can also
    be extended with a path containing a non-tree edge $(x,z)$ such
    that $z \in T_v$ and (a) $x = u$ or (b) $x \in T_w$ and $w$ is
    sibling of $v$; and $d_{G'}(s,x) + w(x,z) + d_{G'}(z,t) \leq
    \alpha$, where $G' = G - (u,v)$.}  \label{fig:weighted:lemma}
  \vspace{-0.35cm}  
\end{figure}

\begin{lemma} \label{lem:lcp}
 Let $\pi_{su} = (s = v_0, v_1), \ldots, (v_{l-1}, v_l = u)$ be a
 common prefix of all paths in $\mathcal{P}_{\alpha}(s,t,G) \neq
 \emptyset$ and $T_s$ a shortest path tree rooted at $s$. Then,
 \begin{enumerate}
 \item the path $\pi_{su}(u,v)$ is a common prefix of
   $\mathcal{P}_{\alpha}(s,t,G)$, if there is no edge $(x,z)$ such
   that $d_{G'}(s,x) + w(x,z) + d_{G'}(z,t) \leq \alpha$, where $G' =
   G - (u,v)$, $z \in T_v$, and (a) $x = u$ or (b) $x \in T_w$ with
   $w$ a sibling of $v$ (see Fig.~\ref{fig:weighted:lemma});
 \item $\pi_{su}$ is the longest common prefix of
   $\mathcal{P}_{\alpha}(s,t,G)$, otherwise.
 \end{enumerate}
\end{lemma}

In order to use the characterization of Lemma~\ref{lem:lcp} for the
longest prefix of $\mathcal{P}_{\alpha}(s,t,G)$, we need to
efficiently test the weight condition given in item~1, namely
$d_{G'}(s,x) + w(x,z) + d_{G'}(z,t) \leq \alpha$, where $G' = G -
(u,v)$ and $(u,v)$ belongs to the tree path from $s$ to $t$. We have
that $d_{G'}(s,x) = d_{G}(s,x)$, since $x$ does not belong to the
subtree of $v$ in the shortest path tree $T_s$. Indeed, only the
distances of vertices in the subtree $T_v$ can possibly change after
the removal of the tree edge $(u,v)$. However, in principle we have no
guarantee that $d_{G'}(z,t)$ also remains unchanged: recall that to
maintain the distances from $t$ we need a tree rooted at $t$ not at
$s$. Clearly, we cannot compute the shortest path tree from $t$ for
each $G'$; in the worst case, this would imply the computation of $n$
shortest path trees. For this reason, we need
Lemma~\ref{lem:path:dist}.  It states that, in the specific case of
the vertices $z$ we need to compute the distance to $t$ in $G'$, we
have that $d_{G'}(z,t) = d_G(z,t)$.

\begin{lemma} \label{lem:path:dist}
  Let $T_s$ be a shortest path tree rooted at $s$ and $t$ a vertex of
  $G$. Then, for any edge $(u,v)$, with $v$ closer to $t$, in the
  shortest path $\pi_{st}$ in the tree $T_s$, we have that $d_G(z,t) =
  d_{G'}(z,t)$, where $z \in T_v$ and $G' = G - (u,v)$.
\end{lemma}

It is not hard to verify that Lemma~\ref{lem:lcp} is also valid for
directed graphs. However, the non-negative hypothesis for the weights
is necessary; more specifically, we need the monotonicity property for
path weights which states that for any path the weight of any subpath
is not greater than the weight of the full path. Now, in
Lemma~\ref{lem:path:dist} both the path monotonicity property and the
fact that the graph is undirected are necessary. Since these two
lemmas are the basis for the efficiency of Algorithm~\ref{alg:lcp}, it
seems difficult to extend it to directed graphs.

Algorithm~\ref{alg:lcp} implements the strategy suggested by
Lemma~\ref{lem:lcp}. Given a shortest path tree $T_s$ of $G$ rooted at
$s$, the algorithm traverses each vertex $v_i$ in the tree path $s =
v_0 \ldots v_n = t$ from the root $s$ to $t$, and at every step finds
all non-tree edges $(x,z)$ entering the subtree rooted at $v_{i+1}$
from a sibling subtree, \ie a subtree rooted at $w \in N^+(v_i)
\setminus \{v_{i+1}\}$. For each non-tree $(x,z)$ linking the sibling
subtrees found, it checks if it satisfies the weight condition
$d_{G'}(s,x) + w(x,z) + d_{G'}(z,t) \leq \alpha$, where $G' = G -
(v_i,v_{i+1})$. Item~2 of the same lemma implies that the first time
an edge $(x,z)$ satisfies the weight condition, the tree path
traversed so far is the longest common prefix of
$\mathcal{P}_{\alpha}(s,t,G)$. In order to test the weight conditions,
as stated previously, we have that $d_{G'}(s,x) = d_G(s,x)$, since $x$
does not belong to the subtree of $v$ in $T_s$. In addition,
Lemma~\ref{lem:path:dist} guarantees that $d_{G'}(z,t) =
d_G(z,t)$. Thus, it is sufficient for the algorithm to compute only
the shortest path trees from $t$ and from $s$ in $G$.

\begin{algorithm} 
\caption{$\lcp(s, t, \alpha, G)$} \label{alg:lcp}
compute $T_s$, a shortest path tree from $s$ in $G$ \\
compute $T_t$, a shortest path tree from $t$ in $G$ \\
let $\pi_{st} = (s = v_0, v_1) \ldots (v_{n-1}, v_n = t)$ be the shortest path in $T_s$ \\
\For{$v_i \in \{v_1, \ldots, v_n\}$} { \label{alg:lcp:for}
  \For{$w \in N^+(v_i) \setminus \{v_{i+1}\}$}{
    let $T_w$ be the subtree of $T_s$ rooted at $w$ \\
    \For{$(x,z) \in G$ s.t. $x \in T_w$ {\bf or} $x = v_i$}{
      \If{$z \in T_{v_{i+1}}$ {\bf and}  $d_G(s,x) + w(x,z) + d_G(z,t) \leq \alpha$}{ \label{alg:lcp:edges}
        {\bf break} \\
      }
    }
  }
}
{\bf return} $\pi_{sv_{i-1}}$
\end{algorithm}

\begin{theorem} \label{teo:lcp}
  Algorithm~\ref{alg:lcp} finds the longest common prefix of
  $\mathcal{P}_{\alpha}(s,t,G)$ in $O(m+t(n,m))$ time using $O(m)$ space.
\end{theorem}
\begin{proof}
  The cost of the algorithm can be divided in two parts: the cost to
  compute the shortest path trees $T_s$ and $T_t$, and the cost of the
  loop in line~\ref{alg:lcp:for}. The first part is bounded by
  $O(t(n,m))$. Let us now prove that the second part is bounded by
  $O(m+n)$. The cost of each execution of line~\ref{alg:lcp:edges} is
  $O(1)$, since we only need distances from $s$ and $t$ and the
  shortest path trees from $s$ and $t$ are already computed, and we
  pre-process the tree to decide in $O(1)$ if a vertex belongs to a
  subtree.  Hence, the cost of the loop is bounded by the number of
  times line~\ref{alg:lcp:edges} is executed. The neighborhood of each
  vertex $x \in T_w$ is visited exactly once, since for each $w \in
  N^+(v_i) \setminus \{v_{i+1}\}$ and $w' \in N^+(v_j) \setminus
  \{v_{j+1}\}$ the subtrees $T_w$ and $T_{w'}$ are disjoint, where
  $v_i$ and $v_j$ belong to the tree path from $s$ to $t$. \qed
\end{proof}

\section{$K$-shortest and $\alpha$-bounded paths: A unified view} \label{sec:weighted:ordered}

The two main differences between the solutions to the $K$-shortest and
$\alpha$-bounded paths problems are: (i) the order in which the paths
are output and (ii) the space complexity of the algorithms. In this
section, we show that both problems can be placed in a unified
framework such that those differences arise in a natural way. More
precisely, we show that their solutions correspond to two different
traversals of the \emph{same} rooted tree: a Dijkstra-like traversal
for the $K$-shortest and a DFS-like traversal for the $\alpha$-bounded
paths. This tree is a weighted version of the recursion tree of
Algorithm~\ref{alg:simple}, so the height is bounded by $n$ and each
leaf corresponds to an $\alpha$-bounded $st$-path (see
Lemma~\ref{lem:rec:tree}).

The space complexity of the algorithms then follows from the fact
that, in addition to the memory to store the tree, Dijkstra's
algorithm uses memory proportional to the number of nodes, whereas the
DFS uses memory proportional to the height of the tree. In addition,
the order in which the solutions are output is precisely the order in
which the leaves of the tree are visited, a Dijkstra-like traversal
visits the leaves in increasing order of their distance from the root,
whereas a DFS-like traversal visits them in an arbitrary but fixed
order.

We first modify Algorithm~\ref{alg:simple} to obtain an iterative
\emph{generic} variant.  The pseudocode is shown in
Algorithm~\ref{alg:iterative}. Observe that each node in the recursion
tree of Algorithm~\ref{alg:simple} corresponds to some tuple $\langle
u, t, \pi_{ut}, G'\rangle$ in line~\ref{alg:iterative:elem} of
Algorithm~\ref{alg:iterative}. By generic we mean that the container
$Q$ is not specified in the pseudocode, the only requirement is the
support for two operations: \emph{push}, to insert a new element in
$Q$; and \emph{pop}, to remove and return an element of $Q$. It should
be clear now that depending on the container, the algorithm will
perform a different traversal in the underlying recursion tree of
Algorithm~\ref{alg:simple}.

\begin{algorithm}[htbp]
\caption{$\listpathsiter(u,t, \alpha, \pi_{su}, G)$}  \label{alg:iterative}   
push $\langle s, t, \emptyset, G\rangle$ in $Q$ \\
\While{$Q$ is not empty} {
  $\langle u, t, \pi_{su}, G'\rangle = Q.pop()$ \\ \label{alg:iterative:elem}
  \uIf{$u = t$}{ 
    output($\pi_{su}$) \\ 
  }
  \Else {
    compute a shortest path tree $T_t$ from $t$ in $G^R - u$ \\
    \For{$v \in N^+(u)$}{
      \If{$d(v,t) \leq \alpha - w(u,v)$}{ \label{alg:iterative:test}
        push $\langle v,t, \alpha - w(u,v), \pi_{su} \cdot (u,v), G' - u \rangle$ in $Q$ 
      } 
    }
  }
}
\end{algorithm}

Algorithm~\ref{alg:iterative} uses the same strategy to partition the
solution space (Eq.~\ref{eq:path:partition}). Of course, the order in
which the partitions are explored depends on the type of container
used for $Q$. We show that if $Q$ is a stack, then the solutions are
output in the reverse order of Algorithm~\ref{alg:simple} and the
maximum size of the stack is linear in the size of the input.  If on
the other hand, $Q$ is a priority queue, using a suitable key, the
solutions are output in increasing order of their lengths, but in this
case the maximum size of the priority queue is linear in the number of
solutions, which is not polynomial in the size of the input.

Let $\mathcal{T}$ be the recursion tree of Algorithm~\ref{alg:simple}
(see Lemma~\ref{lem:rec:tree}). In Algorithm~\ref{alg:iterative}, each
element $\langle u, t, \pi_{su}, G'\rangle$ corresponds to the
arguments of a call of Algorithm~\ref{alg:simple}, \ie a node of
$\mathcal{T}$. For any container $Q$ supporting push and pop
operations, Algorithm~\ref{alg:iterative} visits each node of
$\mathcal{T}$ exactly once, since at every iteration a node from $Q$
is deleted and its children are inserted in $Q$, and $\mathcal{T}$ is
a tree.  In particular, this guarantees that every leaf of
$\mathcal{T}$ is visited exactly once, thus proving the following
lemma.

\begin{lemma}
  Algorithm~\ref{alg:iterative} outputs all $\alpha$-bounded
  $st$-paths.
\end{lemma}

Let us consider the case where $Q$ is a stack. It is not hard to prove
that Algorithm~\ref{alg:simple} is a DFS traversal of $\mathcal{T}$
starting from the root, while Algorithm~\ref{alg:iterative} is an
\emph{iterative} DFS~\cite{Sedgewick01} traversal of $\mathcal{T}$
also starting from the root. Basically, an iterative DFS keeps the
vertices of the fringe of the non-visited subgraph in a stack, at each
iteration the next vertex to be explored is popped from the stack, and
recursive calls are replaced by pushing vertices in the stack. Now,
for a fixed permutation of the children of each node in $\mathcal{T}$,
the nodes visited in an iterative DFS traversal are in the reverse
order of the nodes visited in a recursive DFS traversal, thus proving
Lemma~\ref{lem:weighted:dfs_order}.

\begin{lemma} \label{lem:weighted:dfs_order}
  If $Q$ is a stack, then Algorithm~\ref{alg:iterative} outputs the
  $\alpha$-bounded $st$-path in the reverse order of
  Algorithm~\ref{alg:simple}.
\end{lemma}

For any rooted tree, at any moment during an iterative DFS traversal,
the number of nodes in the stack is bounded by the sum of the degrees
of the root-to-leaf path currently being explored. Recall that every
leaf in $\mathcal{T}$ corresponds to a path in
$\mathcal{P}_{\alpha}(s,t,G)$. Actually, there is a one-to-one
correspondence between the nodes of a root-to-leaf path $P$ in
$\mathcal{T}$ and the vertices of the $\alpha$-bounded $st$-path $\pi$
associated to that leaf. Hence, the sum of the degrees of the nodes of
$P$ in $\mathcal{T}$ is equal to the sum of the degrees of the
vertices $\pi$ in $G$, which is bounded by $m$, thus proving
Lemma~\ref{lem:weighted:dfs_size}. 

\begin{lemma} \label{lem:weighted:dfs_size}
  The maximum number of elements in the stack of
  Algorithm~\ref{alg:iterative} over all iterations is bounded by $m$.
\end{lemma}

Let us consider now the case where $Q$ is a priority queue. There is a
one-to-many correspondence between arcs in $G$ and arcs in
$\mathcal{T}$, \ie if $\mathcal{P}_{\alpha''}(v,t,G'')$ is a child of
$\mathcal{P}_{\alpha'}(u,t,G')$ in $\mathcal{T}$ then $(u,v)$ is an
arc of $G$. For every arc of $\mathcal{T}$, we give the weight of the
corresponding arc in $G$. Now, Algorithm~\ref{alg:iterative} using a
priority queue with $w(\pi_{su}) + d_G(u,t)$ as keys performs a
Dijkstra-like traversal in this weighted version of $\mathcal{T}$
starting from the root. Indeed, for a node $\langle u, t, \pi_{su},
G\rangle$ the distance from the root is $w(\pi_{su})$, and $d_G(u,t)$
is a (precise) estimation of the distance from $\langle u, t,
\pi_{su}, G\rangle$ to the closest leaf of $\mathcal{T}$. In other
words, it is an $A^*$ traversal~\cite{Dechter85} in the weighted
rooted tree $\mathcal{T}$, using the (optimal) heuristic
$d_G(u,t)$. As such, Algorithm~\ref{alg:iterative} explores first the
nodes of $\mathcal{T}$ leading to the cheapest non-visited leaf. This
is formally stated in Lemma~\ref{lem:heap_order}.

\begin{lemma} \label{lem:heap_order}
  If $Q$ is a priority queue with $w(\pi_{su}) + d_{G'}(u,t)$ as the
  priority key of $\langle u, t, \pi_{su}, G'\rangle$, then
  Algorithm~\ref{alg:iterative} outputs the $\alpha$-bounded
  $st$-paths in increasing order of their lengths.
\end{lemma}

For any choice of the container $Q$, each node of $\mathcal{T}$ is
visited exactly once, that is, each node of $\mathcal{T}$ is pushed at
most once in $Q$. This proves Lemma~\ref{lem:heap_size}.

\begin{lemma} \label{lem:heap_size}
  The maximum number of elements in a priority queue of
  Algorithm~\ref{alg:iterative} over all iterations is bounded by
  $\gamma$.
\end{lemma}

Algorithm~\ref{alg:iterative} uses $O(m\gamma)$ space since for every
node inserted in the priority queue, we also have to store the
corresponding graph. Moreover, using a binary heap as a priority
queue, the push and pop operations can be performed in $O(\log
\gamma)$ each, where $\gamma$ is the maximum size of the
heap. Therefore, combining this with Lemma~\ref{lem:heap_order}, we
obtain the following theorem.

\begin{theorem}
  Algorithm~\ref{alg:iterative} using a binary heap outputs all
  $\alpha$-bounded $st$-paths in increasing order of their lengths in
  $O((nt(n,m) + \log \gamma)\gamma)$ total time, using $O(m\gamma)$
  space.
\end{theorem}

{\footnotesize
\bibliographystyle{plain}

}

\end{document}